\documentclass[12pt]{article}%

\IfFileExists{sariel_computer.sty}{\def\sarielComp{1}}{}
\ifx\sarielComp\undefined%
\newcommand{\SarielComp}[1]{}
\newcommand{\NotSarielComp}[1]{#1}%
\else
\newcommand{\SarielComp}[1]{#1}%
\newcommand{\NotSarielComp}[1]{}%
\fi
\newcommand{\IfPrinterVer}[2]{#2}%

\newcommand{\UsePackage}[1]{%
  \IfFileExists{../styles/#1.sty}{%
      \usepackage{../styles/#1}%
   }{%
      \IfFileExists{./styles/#1.sty}{%
         \usepackage{styles/#1}%
      }{%
         \usepackage{#1}%
      }%
   }%
}

\usepackage[cm]{fullpage}%
\usepackage{amsmath}%
\usepackage{amssymb}%
\usepackage{bm}%
\usepackage{mathcalb}%
\usepackage{xcolor}%
\usepackage{upgreek}%
\usepackage{euscript}%
\usepackage{caption}%
\UsePackage{picins}%
\usepackage[normalem]{ulem}
\usepackage{mathcalb}%

\SarielComp{\usepackage{sariel_colors}}%

\usepackage[amsmath,thmmarks]{ntheorem}%
\theoremseparator{.}%

\usepackage{titlesec}%
\titlelabel{\thetitle. }%
\usepackage{xcolor}%
\usepackage{mleftright}%
\usepackage{xspace}%
\usepackage{scalerel}%

\usepackage[inline]{enumitem}

\newlist{compactenumA}{enumerate}{5}%
\setlist[compactenumA]{topsep=0pt,itemsep=-1ex,partopsep=1ex,parsep=1ex,%
   label=(\Alph*)}%

\newlist{compactenumi}{enumerate}{5}%
\setlist[compactenumi]{topsep=0pt,itemsep=-1ex,partopsep=1ex,parsep=1ex,%
   label=(\roman*)}%

\usepackage{ifluatex}
\usepackage{ifxetex}

\ifluatex 
      \usepackage{fontspec}
      \usepackage[utf8]{luainputenc}
\else
       \ifxetex 
          \usepackage{fontspec}
       \else
          \usepackage[T1]{fontenc}
          \usepackage[utf8]{inputenc}
       \fi
\fi
\providecommand{\BibLatexMode}[1]{}
\providecommand{\BibTexMode}[1]{#1}

\ifx\UseBibLatex\undefined%
  \renewcommand{\BibLatexMode}[1]{}
  \renewcommand{\BibTexMode}[1]{#1}
\else
  \renewcommand{\BibLatexMode}[1]{#1}
  \renewcommand{\BibTexMode}[1]{}
\fi

\BibLatexMode{%
   \usepackage[bibencoding=ascii,style=alphabetic,backend=biber]{biblatex}%
   \UsePackage{sariel_biblatex}%
}

\usepackage{hyperref}%

\IfPrinterVer{%
   \usepackage{hyperref}%
}{%
   \usepackage{hyperref}%
   \hypersetup{%
      breaklinks,%
      ocgcolorlinks, colorlinks=true,%
      urlcolor=[rgb]{0.25,0.0,0.0},%
      linkcolor=[rgb]{0.5,0.0,0.0},%
      citecolor=[rgb]{0,0.2,0.445},%
      filecolor=[rgb]{0,0,0.4},
      anchorcolor=[rgb]={0.0,0.1,0.2}%
   }
}

\definecolor{blue25}{rgb}{0,0,0.7}
\providecommand{\emphic}[2]{%
   \textcolor{blue25}{%
      \textbf{\emph{#1}}}%
   \index{#2}}
\providecommand{\emphi}[1]{\emphic{#1}{#1}}

\theoremseparator{.}%

\theoremstyle{plain}%
\newtheorem{theorem}{Theorem}[section]

\newtheorem{lemma}[theorem]{Lemma}

\newtheorem{corollary}[theorem]{Corollary}

\theoremstyle{plain}%
\theoremheaderfont{\sf} \theorembodyfont{\upshape}%
\newtheorem*{remark:unnumbered}[theorem]{Remark}%
\newtheorem{remark}[theorem]{Remark}%
\newtheorem{definition}[theorem]{Definition}

\newcommand{\myqedsymbol}{\rule{2mm}{2mm}}

\theoremheaderfont{\em}%
\theorembodyfont{\upshape}%
\theoremstyle{nonumberplain}%
\theoremseparator{}%
\theoremsymbol{\myqedsymbol}%
\newtheorem{proof}{Proof:}%

\newcommand{\atgen}{\symbol{'100}}
\newcommand{\SarielThanks}[1]{\thanks{Department of Computer Science;
      University of Illinois; 201 N. Goodwin Avenue; Urbana, IL,
      61801, USA; {\tt sariel\atgen{}illinois.edu}; {\tt
         \url{http://sarielhp.org/}.} #1}}

\newcommand{\MitchellThanks}[1]{%
   \thanks{%
      Department of Computer Science;
      University of Illinois; 201 N. Goodwin Avenue; Urbana, IL,
      61801, USA; {\tt mfjones2\atgen{}illinois.edu}; {\tt
         \url{http://mfjones2.web.engr.illinois.edu/}.} #1}}

\numberwithin{figure}{section}%
\numberwithin{table}{section}%
\numberwithin{equation}{section}%

\newcommand{\HLink}[2]{\hyperref[#2]{#1~\ref*{#2}}}
\newcommand{\HLinkSuffix}[3]{\hyperref[#2]{#1\ref*{#2}{#3}}}

\newcommand{\figlab}[1]{\label{fig:#1}}
\newcommand{\figref}[1]{\HLink{Figure}{fig:#1}}

\newcommand{\tbllab}[1]{\label{table:#1}}
\newcommand{\tblref}[1]{\HLink{Table}{table:#1}}

\newcommand{\thmlab}[1]{{\label{theo:#1}}}
\newcommand{\thmref}[1]{\HLink{Theorem}{theo:#1}}

\newcommand{\corlab}[1]{\label{cor:#1}}
\newcommand{\corref}[1]{\HLink{Corollary}{cor:#1}}%

\newcommand{\remlab}[1]{\label{rem:#1}}
\newcommand{\remref}[1]{\HLink{Remark}{rem:#1}}%

\newcommand{\lemlab}[1]{\label{lemma:#1}}
\newcommand{\lemref}[1]{\HLink{Lemma}{lemma:#1}}%

\providecommand{\eqlab}[1]{}%
\renewcommand{\eqlab}[1]{\label{equation:#1}}

\newcommand{\Set}[2]{\left\{ #1 \;\middle\vert\; #2 \right\}}

\newcommand{\pth}[2][\!]{\mleft({#2}\mright)}%

\newcommand{\cardin}[1]{\left| {#1} \right|}%

\newcommand{\tldO}{\scalerel*{\widetilde{O}}{j^2}}%

\renewcommand{\Re}{\mathbb{R}}%

\definecolor{OliveGreen}{cmyk}{0.64, 0, 0.95, 0.40 }

\newcommand{\etal}{\textit{et~al.}\xspace}

\providecommand{\Mh}[1]{#1}

\newcommand{\eps}{\varepsilon}

\newcommand{\R}{\mathbb{R}}

\newcommand{\PS}{\Mh{P}}%
\newcommand{\QS}{\Mh{Q}}%
\newcommand{\US}{\Mh{U}}

\newcommand{\SSet}{S}

\newcommand{\pa}{\Mh{p}}%
\newcommand{\pb}{\Mh{q}}%
\newcommand{\ps}{\Mh{s}}%
\newcommand{\pu}{\Mh{u}}%
\newcommand{\pv}{\Mh{v}}%

\newcommand{\Lines}{L}

\newcommand{\ArrX}[1]{\mathcal{A}\pth{#1}}

\newcommand{\Disks}{\mathcal{D}}
\newcommand{\Rects}{\mathcal{R}}
\newcommand{\Rect}{R}
\newcommand{\TX}[1]{\mathsf{t}\pth{#1}}
\newcommand{\BX}[1]{\mathsf{b}\pth{#1}}
\newcommand{\LX}[1]{\mathsf{l}\pth{#1}}
\newcommand{\RX}[1]{\mathsf{r}\pth{#1}}

\newcommand{\disk}{\mathsf{d}}
\newcommand{\diskB}{\mathsf{q}}
\newcommand{\rad}{r}
\newcommand{\cen}{\mathsf{c}}

\newcommand{\Graph}{G}

\newcommand{\Vertices}{V}
\newcommand{\Edges}{E}
\newcommand{\PX}[1]{{{#1}^+}}
\newcommand{\MX}[1]{{{#1}^-}}
\newcommand{\match}{M}
\newcommand{\ple}{\prec}

\newcommand{\POX}[1]{\pth{#1, \prec}}
\newcommand{\CC}{\mathcal{C}}
\newcommand{\ac}{D}%

\newcommand{\DS}{\mathcalb{D}}
\newcommand{\DSX}[1]{\DS\pth{#1}}
\newcommand{\edistY}[2]{\left\| #1 - #2 \right\|}
\newcommand{\wdistX}[1]{\delta_{#1}}
\newcommand{\wdistY}[2]{\wdistX{#1}\pth{#2}}

\newlength{\savedparindent}
\newcommand{\SaveIndent}{\setlength{\savedparindent}{\parindent}}
\newcommand{\RestoreIndent}{\setlength{\parindent}{\savedparindent}}

\newcommand{\Term}[1]{\textsf{#1}}
\newcommand{\BFS}{\Term{BFS}\xspace}%
\newcommand{\DAG}{\Term{DAG}\xspace}%

\usepackage{array}
\BibLatexMode{%
   \bibliography{antichain}
}

\title{%
   Some Geometric Applications of Anti-Chains
}%
\date{\today}

\author{%
   Sariel Har-Peled%
   \SarielThanks{}%
   \and%
   Mitchell Jones%
   \MitchellThanks{}%
}

\begin{document}

\maketitle


\begin{abstract}
    We present an algorithmic framework for computing anti-chains of
    maximum size in geometric posets. Specifically, posets in which
    the entities are geometric objects, where comparability of two
    entities is implicitly defined but can be efficiently
    tested. Computing the largest anti-chain in a poset can be done in
    polynomial time via maximum-matching in a bipartite graph, and
    this leads to several efficient algorithms for the following
    problems, each running in (roughly) $O(n^{3/2})$ time:
    \begin{compactenumA}
        \item Computing the largest Pareto-optimal subset of a set of $n$
        points in $\Re^d$.

        \smallskip%
    	\item Given a set of disks in the plane, computing the largest
    	subset of disks such that no disk contains another. This is
    	quite surprising, as the independent version of this problem
    	is computationally hard.

        \smallskip%
        \item Given a set of axis-aligned rectangles, computing the
        largest subset of non-crossing rectangles.
    \end{compactenumA}
\end{abstract}


\section{Introduction}


\paragraph*{Partial orderings.}

Let $\POX{\Vertices}$ be a \emphi{partially ordered set} (or a poset),
where $\Vertices$ is a set of entities.  An \emphi{anti-chain} is a
subset of elements $\ac \subseteq \Vertices$ such that all pairs of
elements in $\ac$ are incomparable in $\POX{\Vertices}$. A \emphi{chain}
is a subset $C \subseteq V$ such that all pairs of elements in $C$
are comparable. A chain \emph{cover} $\CC$ is a
collection of chains whose union covers $\Vertices$. Observe that any
anti-chain can contain at most one element from any given chain. As
such, if $\CC$ is the smallest collection of chains covering
$\Vertices$, then for any anti-chain $\ac$, $\cardin{\ac} \leq
\cardin{\CC}$. Dilworth's Theorem \cite{d-dtpos-50} states that the
minimum number of chains whose union covers $\Vertices$ is equal to
the anti-chain of maximum size.


\paragraph*{Implicit posets arising from geometric problems.}

We are interested in implicitly defined posets, where the elements of
the poset are geometric objects. In particular, if one can compute the
largest anti-chain in these implicit posets, we obtain algorithms
solving natural geometric problems.  To this end, we describe a
framework for computing anti-chains in an implicitly defined poset
$\POX{\Vertices}$, under the following two assumptions:
\begin{enumerate*}[label=(\roman*)]
    \item comparability of two elements in the poset can be
    efficiently tested, and
    \item given an element $v \in \Vertices$, one can quickly find an
    element $u \in \Vertices$ with $v \ple u$.
\end{enumerate*}

As an example, let $\PS$ be a set of $n$ points in the plane. Form the
partial ordering $\POX{\PS}$, where $\pb \ple \pa$ if $\pa$ dominates
$\pb$. One can efficiently test comparability of two points, and given
$\pb$, can determine if it is dominated by a point $\pa$ by reducing
the problem to an orthogonal range query. Observe that an anti-chain
in $\POX{\PS}$ corresponds to a collection of points in which no point
dominates another. The largest such subset can be computed efficiently
by finding the largest anti-chain in $\POX{\PS}$, see \lemref{pareto}.

\paragraph*{Previous work.}

Posets have been previously utilized and studied in computational
geometry \cite{sk-gap-98,fw-mkcpps-98,mw-gscpt-92}. For the poset
$\POX{\PS}$ described above, Felsner and Wernisch \cite{fw-mkcpps-98}
study the problem of computing the largest subset of points which can
be covered by $k$-antichains.


\paragraph*{Our results.}

We describe a general framework for computing anti-chains in posets
defined implicitly, see \thmref{compute:ac:ds}.  As a consequence, we
have the following applications: \medskip%
\begin{compactenumA}
    \item \textsf{Largest Pareto-optimal subset.} %
    Let $\PS \subseteq \R^d$ be a set of $n$ points. A point $\pa \in
    \R^d$ \emphi{dominates} a point $\pb \in \R^d$ if $\pa \geq \pb$
    coordinate wise. Compute the largest subset of points $\SSet
    \subseteq \PS$, so that no point in $\SSet$ dominates any other
    point in $\SSet$. In two dimensions this corresponds to computing
    the longest downward ``staircase'', which can be done in $O( n\log
    n)$ time (our algorithm is not interesting in this case).
    However, for three and higher dimensions, it corresponds to a
    surface of points that form the largest Pareto-optimal subset of
    the given point set.

    \smallskip%
    \item \textsf{Largest loose subset.} %
    Let $\Disks$ be a set of $n$ regions in $\R^d$. A subset
    $\SSet \subseteq \Disks$ is \emphi{loose} if for every pair
    $\disk_1, \disk_2 \in S$, $\disk_1 \not\subseteq \disk_2$ and
    $\disk_2 \not\subseteq \disk_1$. This is a weaker concept than
    \emph{independence}, which requires that no pair of objects
    intersect. Surprisingly, computing the largest loose set
    can be done in polynomial time, as it reduces to finding the
    largest anti-chain in a poset. Compare this to the independent
    set problem, which is NP-Hard for all natural shapes in the
    plane (triangles, rectangles, disks, etc).

    \smallskip%
    \item \textsf{Largest subset of non-crossing rectangles.} %
    Let $\Rects$ be a set of $n$ axis-aligned rectangles in the plane.
    Compute the largest subset of rectangles $\SSet \subseteq \Rects$,
    such that every pair of rectangles in $\SSet$ intersect at most
    twice. Equivalently, $\SSet$ is \emphi{non-crossing}, or
    $\SSet$ forms a collection of pseudo-disks.

    \smallskip%
    \item \textsf{Largest isolated subset of points.} %
    Let $\Lines$ be a collection lines in the plane (not necessarily
    in general position), and let $\PS$ be a set of points lying on
    the lines of $\Lines$. A point $\pa \in \PS$ can \emphi{reach} a
    point $\pb \in \PS$ if $\pa$ can travel from left to right, along
    the lines of $\Lines$, to $\pb$. A subset of points $\QS \subseteq
    \PS$ are \emphi{isolated} if no point in $\QS$ can reach any other
    point in $\QS$ using the lines $L$, see \figref{isolated:ex}.
\end{compactenumA}
\medskip%

\begin{figure}
  \hfill%
  \includegraphics[page=2,scale=0.5]{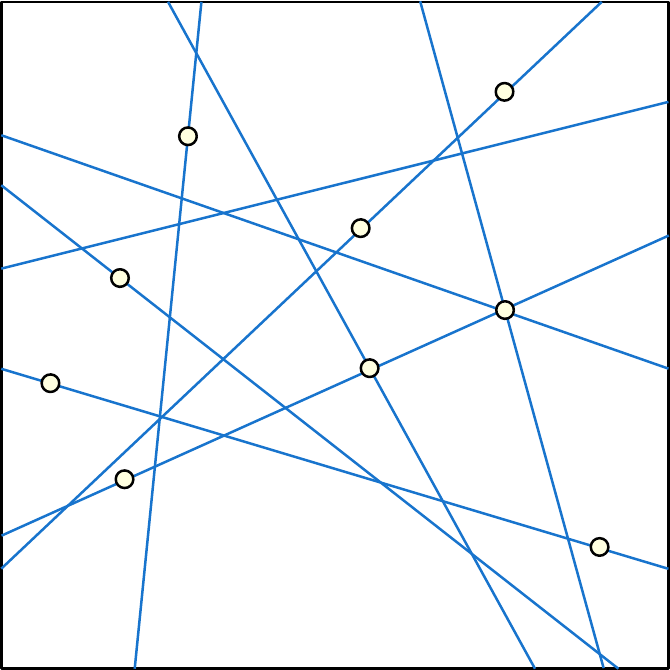}%
  \hfill%
  \includegraphics[page=3,scale=0.5]{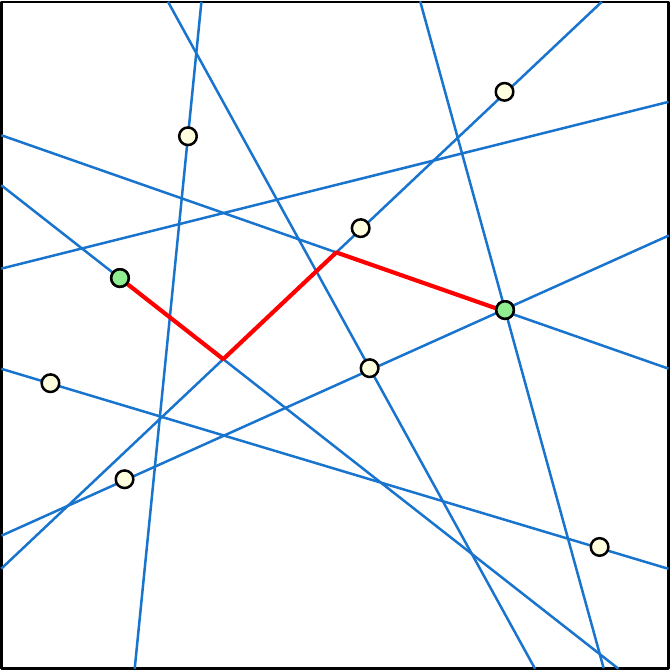}%
  \hfill%
  \includegraphics[page=4,scale=0.5]{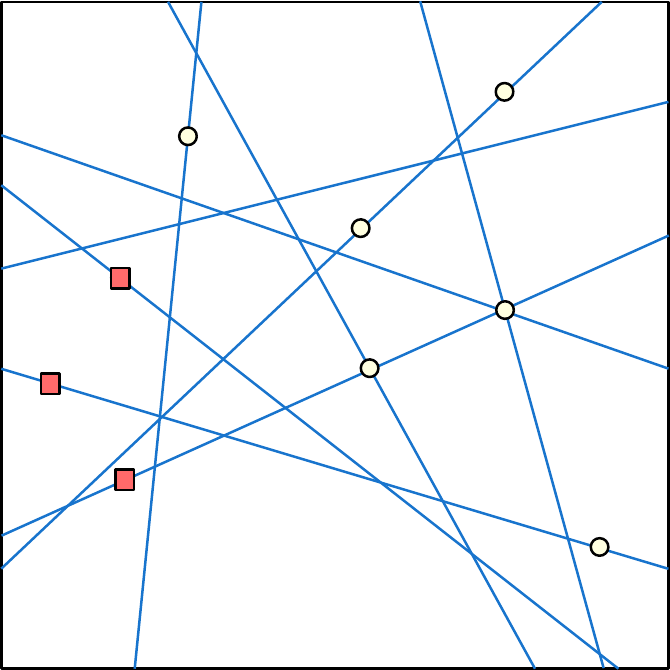}%
  \hfill\phantom{}%
  \caption{Left: A set of points on lines. Middle:
  A point reaching another. Right: An isolated subset.}
  \figlab{isolated:ex}
\end{figure}

Our results are summarized in \tblref{results}.

\begin{table}[t]
  \footnotesize
  \centering
  \begin{tabular}{c|c|c|c}%
    \hline%
    Computing largest subset of
    & Entities
    & Running time
    & Ref\\
    \hline
    \hline
    Pareto-optimal
    & Points in $\R^d$, $d>2$
    & $\tldO(n^{1.5})\Bigr.$
    & \lemref{pareto}
    \\
    \hline
    \hline
    Loose
    & Arbitrary regions in $\Re^d$
    & $O(n^{2.5})\Bigr.$
    & \lemref{loose:any}
    \\%
    \cline{2-4}
    & \begin{minipage}{0.45\linewidth}
        \smallskip%
        \centering
        Arbitrary regions in $\Re^d$ with a
        dynamic range searching data structure,
        $Q(n)$ time per operation
      \end{minipage}
    & $O(n^{1.5}Q(n))\Bigr.$
    & \lemref{loose:ds}
    \\%
    \cline{2-4}
    & Disks in the plane
    & $\tldO(n^{1.5})\Bigr.$
    & \corref{loose:disks}
    \\%
    \hline%
    \hline%
    Non-crossing  %
    & Axis aligned rectangles in $\Re^2$
    & $\tldO(n^{1.5})\Bigr.$
    & \lemref{noncrossing}
    \\%
    \hline
    \hline
    Isolated
    & Points on lines in $\R^2$
    & $O(n^3)$
    & \lemref{isolated}
    \\%
    \hline
  \end{tabular}%
  \caption{Our results, where $\widetilde{O}$ hides factors of the
  form $\log^c n$ ($c$ may depend on $d$).}
  \tbllab{results}
\end{table}


\section{Framework}


\subsection{Computing anti-chains}

The following is a constructive proof of Dilworth's Theorem from the
max-flow min-cut Theorem, and is of course well known \cite{s-co-03}.
We provide a proof for the sake of completeness.

\begin{lemma}
    \lemlab{compute:ac} Let $\POX{\Vertices}$ be a poset. Assume that
    comparability of two elements can be checked in $O(1)$ time. Then
    a maximum size anti-chain in $\POX{\Vertices}$ can be computed in
    $O(n^{2.5})$ time.
\end{lemma}
\begin{proof}
    Given $\POX{\Vertices}$, construct the bipartite graph
    $\Graph = (\US, \Edges)$, where
    $\US = \MX{\Vertices} \cup \PX{\Vertices}$ and
    $\MX{\Vertices}$, $\PX{\Vertices}$ are copies of $\Vertices$. Add
    an edge $(\MX{v}, \PX{u})$ to $\Edges$ when $v \ple u$ in
    $\POX{\Vertices}$. Next, compute the maximum matching in $\Graph$
    using the algorithm of Hopcroft-Karp \cite{hk-ammbg-73}, which
    runs in time
    $O\bigl( \cardin{\Edges} \sqrt{\cardin{\US}} \bigr) = O(n^{2.5})$.
    Let $\match \subseteq \Edges$ be the resulting maximum matching in
    $\Graph$. Define $\MX{\QS} \subseteq \MX{\US}$ as the set of
    unmatched vertices. A path in $\Graph$ is \emph{alternating} if
    the edges of the path alternate between matched and unmatched
    edges. Let $S \subseteq \MX{\Vertices} \cup \PX{\Vertices}$ be the
    set of vertices which are members of alternating paths starting
    from any vertex in $\MX{\QS}$. Finally, set
    $\ac = \Set{v \in \Vertices}{\MX{v} \in S, \PX{v} \not\in S}$.  We
    claim $\ac$ is an anti-chain of maximum size.

    Conceptually, suppose that $\Graph$ is a directed network flow
    graph. Modify $\Graph$ by adding two new vertices $s$ and $t$ and
    add the directed edges $(s, \MX{v})$ and $(\PX{v}, t)$, each with
    capacity one for all $v \in \Vertices$. Finally, direct all edges
    from $\MX{\Vertices}$ to $\PX{\Vertices}$ with infinite
    capacity. By the max-flow min-cut Theorem, the maximum matching
    $\match$, induces a minimum $s$-$t$ cut, which is the cut
    $(s+S, t + \US \setminus S)$, where $S$ is defined above.
    Indeed, $s+S$ is the reachable set from $s$ in the residual
    graph for $\Graph$.
    To see why $\ac$ is an anti-chain, suppose that there
    exist two comparable elements $v, u \in \ac$.  This implies that
    $\MX{v}, \MX{u} \in S$ and $\PX{v}, \PX{u} \notin S$.  Assume without
    loss of generality that $v \ple u$. This implies that
    $(\MX{v}, \PX{u})$ is an edge of the network flow graph $\Graph$
    with infinite capacity that is in the cut
    $(s+S, \Vertices \setminus S + t)$. This contradicts the
    finiteness of the cut capacity.

    \smallskip%

    We next prove that $\cardin{\ac}$ is maximum. Note that an element
    $v \in \Vertices$ is not in $\ac$ if $\MX{v} \notin S$ or
    $\PX{v} \in S$. If $\MX{v} \notin S$ then $(s, \MX{v})$ is in the
    cut. Similarly, if $\PX{v} \in S$ then $(\PX{v}, t)$ is in
    the cut.  Since the minimum cut has capacity $\cardin{M}$, there
    are at most $\cardin{M}$ such vertices, which implies that
    $\cardin{\ac} \geq n - \cardin{M}$.

    On the other hand, a chain cover $\CC$ for $\POX{\Vertices}$ can be
    constructed from $\match$. Given $\POX{\Vertices}$, create
    a \DAG $H$ with vertex set $\Vertices$. We add
    the directed edge $(u, v)$ to $H$ when $u \ple v$.%
    \footnote{Equivalently, $H$ is the transitive closure of the Hasse
    diagram for $\POX{\Vertices}$.}
    Now an edge $(\MX{u}, \PX{v})$ in $\match$ corresponds to the edge
    $(u,v)$ of $H$. As such, a matching corresponds to a collection of
    edges in $H$, where every vertex appears at most twice. Since $H$
    is a \DAG, it follows that $\match$ corresponds to a collection of
    paths in $H$. The end vertex $x$ of such a path corresponds to a
    vertex $\MX{x} \in \MX{\Vertices}$ that is unmatched (as
    otherwise, the path can be extended), and this is the only
    unmatched vertex on this path.
    There are at most $n - \cardin{\match}$ unmatched vertices in
    $\MX{\Vertices}$, which implies that
    $\cardin{\CC} \leq n - \cardin{\match}$. Hence, $\ac$ is an
    anti-chain with
    $\cardin{\ac} \geq n - \cardin{\match} \geq \cardin{\CC}$.
    Additionally, recall that for any anti-chain $\ac'$,
    $\cardin{\ac'} \leq \cardin{\CC}$. These two inequalities imply
    that $\ac$ is of maximum possible size.
\end{proof}

\begin{remark}
    \remlab{compute:chains}%
    As described above, the edges of the matching $\match$ correspond
    to a collection of edges in the \DAG $H$. These edges together
    form a collection of vertex-disjoint paths which cover the
    vertices of $H$, and this is the minimum possible number of paths
    needed to cover the vertices.
\end{remark}


\subsection{Computing anti-chains on implicit posets}

Here, we focus on computing anti-chains in posets, in which
comparability of two elements are efficiently computable. Our main
observation is that one can use range searching data structures to
run the Hopcroft-Karp bipartite matching algorithm faster
\cite{hk-ammbg-73}. This observation goes back to the work of
Efrat \etal \cite{eik-ghbmrp-01}, where they study the problem
of computing a perfect matching $M$ in a weighted bipartite graph $G$
such that the maximum weight edge in $M$ is minimized. They focus on
solving the decision version of the problem: given a parameter $r$, is
there a perfect matching $M$ with maximum edge weight at most $r$?

\begin{theorem}[{{\cite[Theorem 3.2]{eik-ghbmrp-01}}}]
  Let $G$ be a bipartite graph on $n$ vertices with bipartition
  $A \cup B$, and $r > 0$ a parameter.
  For any subset $U \subseteq B$ of $m$ vertices, suppose one can
  construct a data structure $\DSX{B}$ such that:
  \begin{compactenumi}
      \item Given a query vertex $v \in A$, $\DSX{B}$ returns a
      vertex $u \in B$ such that the wight of the edge $(u,v)$ is at
      most $r$ (or reports that no such element in $B$ exists) in
      $T(m)$ time.

      \item An element of $B$ can be deleted from $\DSX{B}$ in
      $T(m)$ time.

      \item $\DSX{B}$ can be constructed in
      $O\bigl(m \cdot T(m)\bigr)$ time.
  \end{compactenumi}
  Then one can decide if there is a perfect matching $M$ in $G$,
  such that all edges in $M$ have weight at most $r$, in
  $O\bigl(n^{1.5} \cdot T(n)\bigr)$ time.
\end{theorem}

Recently, Cabello and Mulzer \cite{cm-mcgig-20} also use a similar
framework as described above for computing minimum cuts in disk graphs
in $\tldO(n^{1.5})$ time. We show that the above framework can also be
extended to computing anti-chains, with a small modification to the
data structure requirements.

\begin{theorem}
    \thmlab{compute:ac:ds}%
    Let $\POX{\Vertices}$ be a poset, where $n =
    \cardin{\Vertices}$. For any subset $\PS \subseteq \Vertices$ of
    $m$ elements, suppose one can construct a data structure
    $\DSX{\PS}$ such that:
    \begin{compactenumi}
        \item Given a query $v \in V$, $\DSX{\PS}$ returns an element
        $u \in \PS$ with $v \ple u$ (or reports that no such element
        in $\PS$ exists) in $T(m)$ time.

        \item An element can be deleted from $\DSX{\PS}$ in $T(m)$
        time.

        \item $\DSX{\PS}$ can be constructed in
        $O\bigl(m \cdot T(m)\bigr)$ time.
    \end{compactenumi}
    Then one can compute the maximum size anti-chain for
    $\POX{\Vertices}$ in $O\bigl(n^{1.5} \cdot T(n)\bigr)$ time.
\end{theorem}
\begin{proof}
    Create the vertex set
    $\US = \MX{\Vertices} \cup \PX{\Vertices}$ of the bipartite
    graph $\Graph = (\US, \Edges)$ associated with
    $\POX{\Vertices}$. The neighborhood of a vertex in the bipartite
    graph can be found by constructing and querying the data structure
    $\DS$. Recall that in each iteration of the maximum matching
    algorithm of Hopcroft-Karp \cite{hk-ammbg-73}, a \BFS tree
    is computed in the residual network of $\Graph$. Such a tree
    can be computed in $O(n T(n))$, as can be easily verified (the \BFS
    algorithm is essentially described below). Furthermore, the
    algorithm terminates after $O(\sqrt{n})$ iterations, which implies
    that one can compute the maximum matching in $\Graph$ in
    $O(n^{1.5} \cdot T(n))$ time. See Efrat \etal \cite{eik-ghbmrp-01}
    for details. Let $\match$ be the matching computed.

    By \lemref{compute:ac}, computing the maximum anti-chain reduces
    to computing the set of vertices which can be reached by
    alternating paths originating from unmatched vertices in
    $\MX{\Vertices}$.  Call this set of vertices $S$, as in
    \lemref{compute:ac}.

    To compute $S$, we do a \BFS in the residual network of $\Graph$.
    To this end, build the data structure $\DSX{\PX{\Vertices}}$.
    Start at an arbitrary unmatched vertex $v \in \MX{\Vertices}$, add
    it to $S$, and query $\DSX{\PX{\Vertices}}$ to travel to a
    neighbor $u \in \PX{\Vertices}$ along an unmatched edge. Add $u$
    to $S$ and delete $u$ from $\DSX{\PX{\Vertices}}$. Travel back to
    a vertex $x$ in $\MX{\Vertices}$ using an edge of the matching
    $\match$ (if possible) and add $x$ to $S$. This process is
    iterated until the alternating path has been exhausted. Then,
    restart the search from $v$ (if $v$ has any remaining unmatched
    neighbors) or another unmatched vertex of $\MX{\Vertices}$.
    Observe that each vertex in $\PX{\Vertices}$ is inserted and
    deleted at most once from the data structure $\DS$. Furthermore,
    each query to $\DS$ can be charged to a vertex deletion. Hence,
    $S$ can be computed in $O(n \cdot T(n))$ time.

    Given $S$, in $O(n)$ time we can compute a maximum anti-chain
    $\ac = \Set{v \in \Vertices}{\MX{v} \in S, \PX{v} \not\in S}$.
\end{proof}


\section{Applications}


\subsection{Largest Pareto-optimal subset of points}

\begin{definition}
  Let $\PS$ be a set of points in $\R^d$. A point $\pa \in \R^d$
  dominates a point $\pb \in \R^d$ if $\pa \geq \pb$ coordinate wise.
  The point set $\PS$ is \emphi{Pareto-optimal} if no point in $\PS$
  dominates any other point in $\PS$.
\end{definition}

\begin{lemma}
    \lemlab{pareto}%
    Let $\PS \subset \R^d$ be a set of $n$ points. A Pareto-optimal
    subset of $\PS$ of maximum size can be computed in
    $O(n^{1.5} (\log n/\log\log n)^{d-1})$ time.
\end{lemma}
\begin{proof}
    Form the implicit poset $\POX{\PS}$ where $\pb \ple \pa$ $\iff$
    $\pa$ dominates $\pb$. Hence, two elements are incomparable when
    neither is dominated by the other. As such, computing the largest
    Pareto-optimal subset is reduced to finding the maximum anti-chain
    in $\POX{\PS}$.

    To apply \thmref{compute:ac:ds}, one needs to exhibit a data
    structure $\DS$ with the desired properties. For a given query
    $\pb$, finding a point $\pa \in \PS$ with $\pb \ple \pa$
    corresponds to finding a point $\pa$ which dominates $\pb$.
    Equivalently, such a point in $\PS$ exists if and only if it lies in
    the range $[\pb_1, \infty) \times \ldots \times [\pb_d, \infty)$.
    This is a $d$-sided orthogonal range query.  Chan and Tsakalidis's
    dynamic data structure for orthogonal range searching
    \cite{ct-dorsrr-17} suffices---their data structure can handle
    deletions and queries in $T(n) = O((\log n/\log\log n)^{d-1})$
    amortized time, and can be constructed in $O(n \cdot T(n))$ time.
\end{proof}

\subsubsection{Chain decomposition}
Let $\PS$ be a set of $n$ points in $\Re^d$. We would like to
decompose $\PS$ into disjoint chains of dominated points, such
that all the points are covered, and the number of chains is
minimum. By \remref{compute:chains}, this can be done by running
the algorithm \lemref{pareto}, and converting the bipartite
matching to chains.  This would take
$O(n^{1.5} (\log n/\log\log n)^{d-1})$ time.

\SaveIndent%
\smallskip%
\noindent%
\begin{minipage}{0.75\linewidth}
    \RestoreIndent%
    As a concrete example, suppose we want to solve a more restricted
    problem in the planar case---decomposing a given point set into
    chains of points (that are monotone in both $x$ and $y$), such
    that no pair of chains intersect.

    Suppose that we have computed a chain decomposition of the points
    $\PS$. Let $C_1$ and $C_2$ be two chains of points, each with an
    edge $\pa_i \pb_i$ in $C_i$ (and $\pa_i$ dominates $\pb_i$) for $i
    = 1, 2$ such that $\pa_1 \pb_1$ and $\pa_2 \pb_2$ intersect in the
    plane.  An exchange argument shows that by deleting these two
    edges and adding the edges $\pa_1 \pb_2$ to $C_1$ and $\pa_2
    \pb_1$ to $C_2$, we decrease the total length of the chains.
    Indeed, let $\ps$ be the intersection point of the edges $\pa_1
    \pb_1$ and $\pa_2 \pb_2$. By the triangle inequality and assuming
    the points of $\PS$ are in general position,
\end{minipage}\hfill
\begin{minipage}{0.25\linewidth}
    \hfill%
    \includegraphics{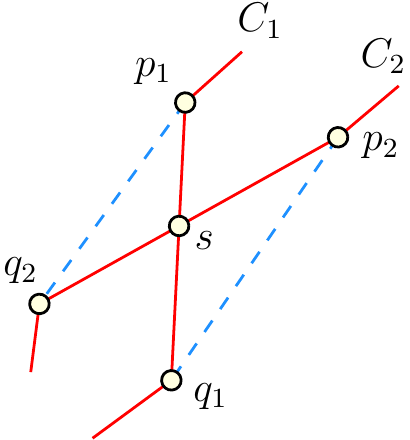}
\end{minipage}%

\begin{align*}
  \| \pa_1 - \pb_2 \| + \| \pa_2 - \pb_1 \|
  < \| \pa_1 - \ps \| + \| \ps - \pb_2 \| + \| \pa_2 - \ps \| + \| \ps - \pb_1 \|
  = \| \pa_1 - \pb_1 \| + \| \pa_2 - \pb_2 \|.
\end{align*}

As such, suppose we assign a weight to each edge in the bipartite
graph $\Graph$ equal to the distance between the two corresponding
points. If $k$ is the size of the maximum (unweighted) matching in
$\Graph$, then we can compute a matching of cardinality $k$ with
minimum weight by solving a min-cost flow instance on $\Graph$ (using
the weights on the edges of $\Graph$ as the costs). This implies that
the resulting chain decomposition covers all points of $\PS$, and no
pair of chains intersect. We obtain the following.

\begin{lemma}
    Let $\PS$ be a set of $n$ points in the plane in general position.
    In polynomial time, one can compute the minimum number of
    \underline{non-intersecting} $(x,y)$-monotone polygonal curves
    covering the points of $\PS$, where every point of $\PS$ must be a
    vertex of one of these polygonal curves, and the vertices of the
    polygonal curves are points of $\PS$.
\end{lemma}

\begin{remark}
  If we do not require the collection of polygonal curves to be
  non-intersecting, there is a much simpler algorithm. Create a
  directed acyclic graph $H = (P, E)$, where $(p,q) \in E$
  if $p$ dominates $q$. Observe that all of the points $S \subseteq P$
  in $H$ with out-degree zero form a polygonal curve. We add this curve
  to our collection and recursively compute the set of curves on the
  residual graph $H \setminus S$. While the number of polygonal
  curves is minimum, the resulting curves may intersect.
\end{remark}


\subsection{Largest loose subset of regions}

\begin{definition}
  Let $\Disks$ be a collection of $n$ regions in $\R^d$.  Such a
  collection $\Disks$ is \emph{loose} if no region of $\Disks$ is
  fully contained inside another region of $\Disks$.
\end{definition}

\begin{lemma}
  \lemlab{loose:any}
  Let $\Disks$ be a collection of $n$ regions in $\R^d$. Suppose
  that for any two regions in $\Disks$, we can test if one
  is contained inside the other in $O(1)$ time. Then the largest
  loose subset of $\Disks$ can be computed in $O(n^{2.5})$ time.
\end{lemma}
\begin{proof}
  Form the implicit poset $\POX{\Disks}$, where $\disk' \ple \disk$
  $\iff$ the region $\disk$ is contained in the interior of $\disk'$.
  In particular, a subset of regions are loose if and only if
  they form an anti-chain in $\POX{\Disks}$.

  By \lemref{compute:ac} (and the assumption that containment
  of objects can be tested in $O(1)$ time) the largest anti-chain,
  and thus the largest loose subset, can be computed in $O(n^{2.5})$
  time.
\end{proof}

\begin{lemma}
  \lemlab{loose:ds}
  Let $\Disks$ be a collection of $n$ regions in $\R^d$.
  For any subset $R \subseteq \Disks$ of
  $m$ regions, suppose one can construct a data structure
  $\DSX{R}$ such that:
  \begin{compactenumi}
      \item Given a query $\disk \in \Disks$, $\DSX{R}$ returns a
      region $\disk' \in R$ with $\disk' \subseteq \disk$
      (or reports that no such region in $R$ exists)
      in $Q(m)$ time.

      \item A region can be deleted from $\DSX{R}$ in $Q(m)$
      time.

      \item $\DSX{R}$ can be constructed in
      $O\bigl(m \cdot Q(m)\bigr)$ time.
  \end{compactenumi}
  Then one can compute the largest loose subset of $\Disks$ in
  $O\bigl(n^{1.5} \cdot Q(n)\bigr)$ time.
\end{lemma}
\begin{proof}
  The proof follows by considering the poset $\POX{\Disks}$
  described in \lemref{loose:any} and applying
  \thmref{compute:ac:ds} using the data structure $\DS$.
\end{proof}

\subsubsection{Largest loose subset of disks}

We show how to compute the largest loose subset when the regions
are disks in the plane. To apply \lemref{loose:ds}, we need to
exhibit the required dynamic data structure $\DS$.

\begin{lemma}
    \lemlab{disks:ds}%
    Let $\Disks$ be a set of $n$ disks in the plane.  There is a
    dynamic data structure $\DS$, which given a query disk $\diskB$,
    can return a disk $\disk' \in \Disks$ such that
    $\disk' \subseteq \diskB$ (or report that so such disk exists) in
    $O(\log^2 n)$ deterministic time. Insertion and deletion of disks
    cost $O(\log^{10 + \eps} n)$ amortized expected time, for all
    $\eps > 0$.
\end{lemma}
\begin{proof}
    Associate each disk $\disk \in \Disks$, which has center
    $\cen_\disk$ and radius $\rad_\disk$, with a weighted distance
    function $\wdistX{\disk} : \R^2 \to \R$, where
    $\wdistY{\disk}{\pa} = \edistY{\cen_\disk}{\pa} + \rad_\disk$.
    Observe that a disk $\disk$ is contained inside the interior of a
    disk $\diskB$ if and only if
    $\wdistY{\disk}{\cen_\diskB} \leq \rad_\diskB$.  For a query disk
    $\diskB$, our goal will be to compute
    $\arg\min_{\disk \in \Disks} \wdistY{\disk}{\cen_\diskB}$. After
    finding such a disk $\disk'$, return that
    $\disk' \subseteq \diskB$ if and only if
    $\wdistY{\disk'}{\cen_\diskB} \leq \rad_\diskB$.

    Hence, the problem is reduced to dynamically maintaining the
    function $F(\pa) = \min_{\disk \in \Disks} \wdistY{\disk}{\pa}$,
    for all $\pa \in \R^2$, under insertions and deletions of disks.
    Equivalently, $F$ is also the lower envelope of the $xy$-monotone
    surfaces defined by $\Set{\wdistX{\disk}}{\disk \in \Disks}$ in
    $\R^3$.  This problem was studied by Kaplan \etal
    \cite{kmrss-dpvdgf-17}: They prove that if $F$ is defined by a
    collection of additively weighted Euclidean distance functions,
    then $F$ can be computed for a given query $\pa$ in $O(\log^2 n)$
    time.  Furthermore, updates can be handled in
    $O\big(2^{O(\alpha(\log n)^2)} \log^{10} n \big)$ time, where
    $\alpha(n)$ is the inverse Ackermann function.
\end{proof}

\begin{corollary}
    \corlab{loose:disks}
    Let $\Disks$ be a set of $n$ disks in the plane. The largest
    loose subset of disks can be computed in
    $O(n^{1.5} \log^{10 + \eps} n)$ expected time, for all $\eps > 0$.
\end{corollary}
\begin{proof}
    Follows from \lemref{loose:ds} in conjunction with the data
    structure described in \lemref{disks:ds}.
\end{proof}

\begin{remark}
    By \remref{compute:chains}, one can decompose a given set of $n$
    disks into the minimum number of disjoint \emph{towers}, where
    each tower is a sequence of disks of the form
    $d_1 \subseteq d_2 \subseteq \cdots \subseteq d_t$.  The resulting
    running time is as stated in \corref{loose:disks}.
\end{remark}


\subsection{Largest subset of non-crossing rectangles}

\begin{definition}
  A collection $\Rects$ of axis-aligned rectangles are non-crossing if
  the boundaries of every pair of rectangles in $\Rects$ intersect at
  most twice.
\end{definition}

\begin{lemma}
    \lemlab{noncrossing}
    Let $\Rects$ be a set of $n$ axis-aligned rectangles in the plane.
    A non-crossing subset of $\Rects$ of maximum size can be computed in
    $O(n^{1.5} (\log n/\log \log n)^3)$ time.
\end{lemma}
\begin{proof}
    For each rectangle $\Rect \in \Rects$, let $\TX{\Rect}$ and
    $\BX{\Rect}$ denote the $y$-coordinate of the top and bottom sides
    of $\Rect$, respectively. Similarly, $\LX{\Rect}$ and $\RX{\Rect}$
    denotes the $x$-coordinate for the left and right sides of $R$.

    Form the poset $\POX{\Rects}$ where
    \begin{align*}
      \Rect' \ple \Rect
      \iff
      [\LX{\Rect}, \RX{\Rect}] \subseteq [\LX{\Rect'}, \RX{\Rect'}]
      \text{ and }
      [\BX{\Rect'}, \TX{\Rect'}] \subseteq [\BX{\Rect}, \TX{\Rect}].
    \end{align*}
    Observe two rectangles $\Rect$ and $\Rect'$ are incomparable if
    and only if the boundaries of $\Rect$ and $\Rect'$ intersect at
    most twice. In particular, the largest subset of non-crossing
    rectangles corresponds to the largest anti-chain in
    $\POX{\Rects}$.

    To apply \thmref{compute:ac:ds}, we need a dynamic data structure
    which, given a rectangle $\Rect \in \Rects$, returns any rectangle
    in $\Rect' \in \Rects$ where $\Rect \ple \Rect'$. Equivalently, we
    want to return a rectangle $\Rect' \in \Rects$ such that
    $[\LX{\Rect'}, \RX{\Rect'}] \subseteq [\LX{\Rect}, \RX{\Rect}]$
    and
    $[\BX{\Rect}, \TX{\Rect}] \subseteq [\BX{\Rect'}, \TX{\Rect'}]$.
    To do so, map each rectangle $\Rect \in \Rects$ to the point
    $(\LX{\Rect}, \RX{\Rect}, \TX{\Rect}, \BX{\Rect}) \in \R^4$. The
    query of interest reduces to a 4-sided orthogonal range query in
    $\R^4$. Chan and Tsakalidis's dynamic data structure for orthogonal
    range searching \cite{ct-dorsrr-17} supports such queries and
    updates in time $O((\log n/\log\log n)^3)$, implying the result.
\end{proof}


\subsection{Largest subset of isolated points}

Let $\Lines$ be a set of $n$ lines in the plane. We assume that no
line in $\Lines$ is vertical and $\Lines$ may not necessarily
be in general position. Let $\PS$ be a set of $n$ points lying on the
lines of $\Lines$.

\begin{definition}
  Given a set of lines $\Lines$ and points $\PS$ lying on $\Lines$,
  a $\pa \in \PS$ can reach a point $\pb \in \PS$ if it possible for
  $\pa$ to reach $\pb$ by traveling from left to right along lines in
  $\Lines$.
  The set $\PS$ is isolated if no point in $\PS$ can reach another
  point in $\PS$.
\end{definition}

\paragraph*{The partial ordering.} Fix the collection of
lines $\Lines$. Given $\PS$, create the poset $\POX{\PS}$,
where $\pa \ple \pb$ $\iff$ $\pa$ can reach $\pb$ using the lines of
$\Lines$. Observe that any subset of isolated points directly
corresponds to an anti-chain in $\POX{\PS}$.

\begin{lemma}
  \lemlab{isolated}
  Let $\PS$ be a collection of $n$ points in the plane
  lying on a set $\Lines$ of $n$ lines. The largest subset of
  isolated points can be computed in $O(n^3)$ time.
\end{lemma}
\begin{proof}
  We can assume that every point of $\PS$ lies on at least two
  lines of $\Lines$. If not, shift such a point $\pa$ to the
  right along the line it lies on, until $\pa$ encounters
  an intersection.

  Start by computing the arrangement $\ArrX{\Lines}$ of the lines
  $\Lines$. Next, construct a directed graph $\Graph$ with vertex set
  equal to the vertices of $\ArrX{\Lines}$. By assumption, $\PS$ is a
  subset of the vertices of $\Graph$. The edges of $\Graph$ consist
  of the edges of the arrangement $\ArrX{\Lines}$ (any edges of
  $\ArrX{\Lines}$ which are half-lines are ignored). For each edge of
  $\ArrX{\Lines}$ with endpoints $\pu, \pv$, we direct the edge in
  $\Graph$ from $\pu$ to $\pv$ when $\pu$ has a smaller $x$-coordinate
  than $\pv$.  Next, for each $\pa \in \PS$, determine the set of
  points of $\PS$ reachable from $\pa$ by performing a \BFS
  in $\Graph$. Thus, given any two points, we can determine if
  they are comparable in $O(1)$ time. Apply \lemref{compute:ac} to
  obtain the largest isolated subset.

  To analyze the running time, note that computing the arrangement
  $\ArrX{\Lines}$ and constructing $\Graph$ can be done in $O(n^2)$
  time.  A \BFS from $n$ points in $\Graph$ costs $O(n^3)$ time total.
  Finally, the largest isolated subset can be found in $O(n^{2.5})$
  time by \lemref{compute:ac}.
\end{proof}




\newpage%
\BibTexMode{%
   \bibliographystyle{alpha}%
   \bibliography{antichain}%
}

\BibLatexMode{\printbibliography}


\end{document}